\documentclass[12pt]{article}
\pdfoutput=1
\usepackage{graphicx}
\usepackage{mathrsfs}
\usepackage{amssymb}
\usepackage{amsmath}
\usepackage{graphicx}
\usepackage{amsmath,amsfonts}
\usepackage{algorithm}
\usepackage{algorithmicx}
\usepackage{algpseudocode}
\usepackage[round]{natbib}
\usepackage{amscd}
\usepackage[mathscr]{eucal}
\usepackage{lineno}
\usepackage{extarrows}

\usepackage{listings}
\usepackage{tikz}

\newtheorem{theorem}{Theorem}[section]
\newtheorem{thm}{Theorem}[section]
\numberwithin{thm}{section}
\newtheorem{remark}[thm]{Remark}
\newtheorem{lemma}[thm]{Lemma}

\newenvironment{proof}{\noindent\\ \noindent\relax{\sc
     Proof}}{{\samepage\par\nopagebreak\hbox
     to\hsize{\hfill$\Box$}}}
\newcommand{\be}{\begin{equation}} \newcommand{\ee}{\end{equation}}
\newcommand{\bd}{\begin{displaymath}} \newcommand{\ed}{\end{displaymath}}
\newcommand{\ba}{\begin{align}} \newcommand{\ea}{\end{align}}
\newcommand{\baa}{\begin{align*}} \newcommand{\eaa}{\end{align*}}
\newcommand{\ben}{\begin{enumerate}} \newcommand{\een}{\end{enumerate}}
\newcommand{\bi}{\begin{itemize}} \newcommand{\ei}{\end{itemize}}

\newcommand{\ud}{\mathrm{d}}
\newcommand{\E}[1]{\operatorname{E}\left[ #1 \right]}

\newcommand{\Expectation}[1]{\operatorname{E}\left[ #1 \right]}

\newcommand{\variance}[1]{\operatorname{Var}\left[ #1 \right]}

\usepackage[normalem]{ulem}

\algnewcommand\And{\textbf{and}}

\begin{document}


\title{Limit distribution of the quartet balance index for Aldous's $\beta\ge 0$--model}
\author{Krzysztof Bartoszek} 

\maketitle

\begin{abstract}
This paper builds up on T. Mart\'inez--Coronado, A. Mir, F. Rossell\'o and G. Valiente's work 
``A balance index for phylogenetic trees based on quartets'', introducing a new 
balance index for trees. We show here that this balance index, in the case of Aldous's $\beta\ge 0$--model,
convergences weakly to a distribution that can be characterized as the fixed point of a contraction operator
on a class of distributions.
\end{abstract}


\section{Introduction}\label{secIntro}
Phylogenetic trees (from a graph theory perspective trees, connected graphs without any cycles,
that have a distinct node, called ``root'', that is interpreted as the ``start'' of the tree)
 are key to evolutionary biology. However, they are not easy to summarize or compare
as it might not be obvious how to tackle their topologies, understood as the internal branching
structure of the trees. 
Therefore, many summary
indices have been proposed in order to ``project'' a tree into $\mathbb{R}$. Such indices have as their
aim to quantify some property of the tree and one of the most studied properties is the symmetry
of the tree. Tree symmetry is commonly captured by a balance index. Multiple balance indices 
have been proposed, Sackin's \cite{MSac1972}, Colless' \cite{DCol1982} or the 
total cophenetic index \cite{AMirFRosLRot2013}. A compact introduction to phylogenetics,
containing in particular a list of tree asymmetry measures (p. $562$--$564$),
can be found in \cite{JFel2004}.
This work accompanies a newly proposed balance index---the quartet index (QI, \cite{TCorAMirFRosGVal2018arXiv}).

One of the reasons for introducing summary indices for trees is to use them for significance testing---does
a tree come from a given probabilistic model. Obtaining the distribution (for a given $n$--number of contemporary
species, i.e. leaves of the tree,
or in the limit $n\to \infty$) of indices is usually difficult and often is done only for the
``simplest'' Yule (pure--birth \cite{GYul1924})  tree case
and sometimes uniform model (see e.g. \cite{DAld1991,MSteAMcK2001}).

Using the contraction method, central limit theorems were found
for various balance indices, like 
the total cophenetic index (Yule model case \cite{KBar2018})
and jointly for Sackin's  and Colless' (in the Yule and uniform model cases \cite{MBluOFraSJan2006}).
Furthermore, in \cite{MBluOFra2006} it was shown that Sackin's index has the same weak limit as the number of
comparisons of the quicksort algorithm \cite{CHoa1962}, both after normalization of course.

In \cite{HChaMFuc2010}  the number of occurrences of patterns in a tree are considered, where
a pattern is understood as ``any subset of the set of all phylogenetic trees of fixed size $k$''. For a
tree with $n$ leaves such a pattern will satisfy the recursion

$$
X_{n,k} \stackrel{\mathcal{D}}{=} X_{L_{n},k} + X^{\ast}_{n-L_{n},k}
$$
where $X_{n,k}$, $X^{\ast}_{n,k}$ and $L_{n}$ are independent,
$X_{n,k}\stackrel{\mathcal{D}}{=}X^{\ast}_{n,k}$ and $L_{n}$  is the size of the left subtree branching from the root.
For the Yule and uniform models they derived central limit theorems (normal limit distribution) with Berry--Esseen 
bounds and Poisson approximations in the total variation distance.
The above description is rather abstract but can be related to in a more direct way. The term $n$ is the number of leaves
of the tree (i.e. nodes of degree $1$). The pattern of fixed size $k$ is a generic term, but 
in Table $1$ in \cite{HChaMFuc2010} concrete examples are given, $k$--pronged nodes, $k$--caterpillars, or nodes
with minimal clade size $k$. In the present manuscript it will be the number of fully balanced subtrees with $k=4$ leaf
nodes. However, in our case the recursion will be of a non--homogeneous form, hence the results from
\cite{HChaMFuc2010} cannot be carried over. The random variable $X_{n,k}$ is the number of occurrences
of the given pattern (of size $k$) in a tree of size $n$. In principle the index $k$ could be dropped at this
description level, but we kept it here for consistency with \cite{HChaMFuc2010}.

Even though the pure--birth model seems to be very widespread in the phylogenetics community,
more complex models need to be studied, especially in the context of tree balance. 
From  Lemma $4$ in \cite{SRocSSni2013} it can be deduced that Yule trees have to be rather balanced---as 
the maximum quartet weight (maximum of number of randomly placed marks along branches
over induced subtrees on four leaves) is asymptotically proportional to the expectation of the tree's height.

In this work here, using the contraction method, we show
convergence in law of the (scaled and centred) quartet index and derive a representation 
(as a fixed point of a particular contraction operator)
of the weak--limit. Remarkably, this is possible not only for the 
Yule tree case but also for Aldous's more general  $\beta$--model (in the $\beta\ge 0$ regime).

The paper is organized as follows. In Section \ref{secPrelim} we 
introduce Aldous's $\beta$--model and the quartet index. 
In Section \ref{secCTdist} 
we prove our main result---Thm. \ref{thmYnYinfConv} via the contraction method. When studying the limit behaviour 
of recursive--type indices for pure--birth binary trees one has that for each internal
node the leaves inside its clade are uniformly split into to sub--clades as the node splits.
However, in Aldous's $\beta$--model this is not the case, the split is according to
a BetaBinomial distribution, and a much finer analysis is required to show 
weak--convergence, with $n$, of the recursive--type index to the fixed point
of the appropriate contraction. Theorem \ref{thmYnYinfConv} is not specific
for the quartet index but covers a more general class of models, where each
internal node split divides its leaf descendants according to a BetaBinomial 
distribution (with $\beta\ge 0$). 
In Section \ref{secLimQI} we apply Thm. \ref{thmYnYinfConv} to the quartet index
and characterize its weak limit. Then, in Section \ref{secSimul} we
illustrate the results with simulations. Finally, in the Appendix we provide R code
used to simulate from this weak limit.

\section{Preliminaries}\label{secPrelim}
\subsection{Aldous's $\beta$--model for phylogenetic trees}
Birth--death models are popular choices for modelling the evolution
of phylogenetic trees. 
However, in \cite{DAld1996,DAld2001}
a different class of models was proposed---the so--called $\beta$--model for binary 
phylogenetic trees.

The main idea behind this model is to consider a (suitable) family $\{ q_{n} \}_{n=2}^{\infty}$
of symmetric, $q_{n}(i)=q_{n}(n-i)$, probability distributions on the natural numbers.
In particular $q_{n} : \{1,\ldots,n-1\} \to [0,1]$. The tree grows in a natural way.
The root node of a $n$--leaf tree defines a partition of the $n$ nodes into two sets of sizes 
$i$ and $n-i$ $(i \in \{1,\ldots,n-1\})$. We randomly choose the number of leaves of the left subtree,
 $L_{n}=i$, according to the distribution $q_{n}$ and this induces the number of leaves, $n-L_{n}$, in the right subtree.
We then repeat recursively in the left and right subtrees, i.e. splitting according to the
distributions $q_{L_{n}}$ and $q_{n-L_{n}}$ respectively. 
Notice that due to $q_{n}$'s symmetry the terms left and right do not have any
particular meaning attached. 

In \cite{DAld1996} it was proposed to consider a one--parameter, $-2 \le \beta \le \infty$, family of probability 
distributions,

\be \label{eqqni}
q_{n}(i) = \frac{1}{a_{n}(\beta)} \frac{\Gamma(\beta+i)\Gamma(\beta+n-i)}{\Gamma(i)\Gamma(n-i)},~~~~1 \le i \le n-1,
\ee
where $a_{n}(\beta)$ is the normalizing constant and $\Gamma(\cdot)$ the Gamma function.
We may actually recognize this as the BetaBinomial$(n-2,\beta+1,\beta+1)$ distribution and represent

\be \label{eqqniBB}
q_{n}(i) = B(\beta+1,\beta+1)^{-1}\int\limits_{0}^{1}\left(\binom{n-2}{i-1}\tau^{i-1}(1-\tau)^{n-i-1}\right)\tau^{\beta}(1-\tau)^{\beta}\ud \tau,
\ee
where $B(a,b)$ is the Beta function with parameters $a$ and $b$. 
Notice that we slightly changed $n$ to $n-2$ and $i$ to $i-1$ in the right side of the equations
with respect to \cite{DAld1996} in order to have better correspondence with the rest of the manuscript here.
Writing informally, from the form of the probability distribution
function, Eq. \eqref{eqqniBB}, we can see that if we would condition under the integral on $\tau$, then
we obtain a binomially distributed random variable. This is a key observation that is the intuition
for the analysis presented here.

Particular values of $\beta$ correspond to some well known models.
The uniform tree model is represented by $\beta=-3/2$, and the pure birth, Yule, model by $\beta=0$.
The limit case of $\beta=\infty$, is $q_{n}(i) \to \binom{n-2}{i-1}2^{-(n-2)}$, i.e. the 
binomial distribution, with success probability equalling $0.5$. This corresponds to the
so--called ``symmetric binary trie'' in computer science literature (e.g. Ch. $5.3$ in \cite{HMah1992})
and was mentioned as the ``random partition tree'' in the evolutionary biology literature
\cite{WMadMSla1991}.

Of particular importance to our work is the limiting behaviour of the scaled
size of the left (and hence right) subtree, $n^{-1}L_{n}$. Lemma $3$ in \cite{DAld1996}
characterizes these asymptotics.

\begin{lemma}[Lemma 3 for $\beta > -1$, {\cite{DAld1996}}]\label{lemLem3Ald}
\begin{enumerate}
\item $\beta=\infty$, $n^{-1}L_{n} \stackrel{\mathcal{D}}{\to} \frac{1}{2}$ ;
\item $-1<\beta<\infty$, $n^{-1}L_{n} \stackrel{\mathcal{D}}{\to} \tau_{\beta}$, where
$\tau_{\beta}$ has the Beta distribution

\be \label{eqBtreedist}
f(x)=\frac{\Gamma(2\beta+2)}{\Gamma^{2}(\beta+1)}x^{\beta}(1-x)^{\beta},~~0<x<1.
\ee
\end{enumerate}
\end{lemma}
\subsection{Quartet index}
In \cite{TCorAMirFRosGVal2018arXiv} a new type of balance
index for discrete (i.e. without branch lengths, or in the language of graph
theory weights assigned to branches) phylogenetic trees---the quartet index.
This index is based on considering the number of so--called quartets of each type made up by the leaves
of the tree. A (rooted) quartet is the induced subtree (a subtree formed by removing
all but some given set of leaves and then removing all, except the root, degree two nodes) from choosing some four leaves.
We should make a point here about the used nomenclature. Usually in the 
phylogenetic literature a quartet is an unrooted tree on four leaves (e.g. \cite{CSemMSte2003}). 
However, here we consider rooted trees and following \cite{TCorAMirFRosGVal2018arXiv}
by (rooted) quartet we mean a rooted tree on four leaves. We will from now on
write quartet for this, dropping the ``rooted'' clarification.

For a given tree $T$, let $\mathcal{P}_{4}(T)$ be the set of quartets of the tree.
Then, the quartet index of $T$ is defined as

\be
QI(T) = \sum\limits_{\mathcal{P}_{4}(T)} QI(Q),
\ee
where $QI(Q)$ assigns a predefined value to a specific quartet (i.e. given tree topology on four leaves). 
When the tree is a binary
one (as here) there are only two possible topologies on four leaves (see Fig. \ref{figTreeB4K4}).
Following \cite{TCorAMirFRosGVal2018arXiv}, Table 1 therein,
we assign the value $0$ to $K_{4}$ quartets and $1$ to $B_{4}$ quartets. Therefore, the
QI for a binary tree (QIB) will be

\be
QIB(T) = \mathrm{number~of~} B_{4} \mathrm{~quartets~of~} T.
\ee

\begin{figure}[!ht]
\centering
\includegraphics[width=0.43\textwidth]{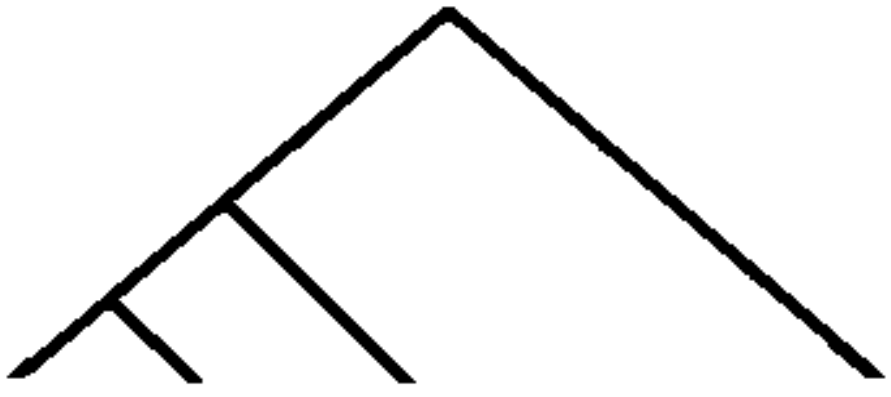}
\includegraphics[width=0.45\textwidth]{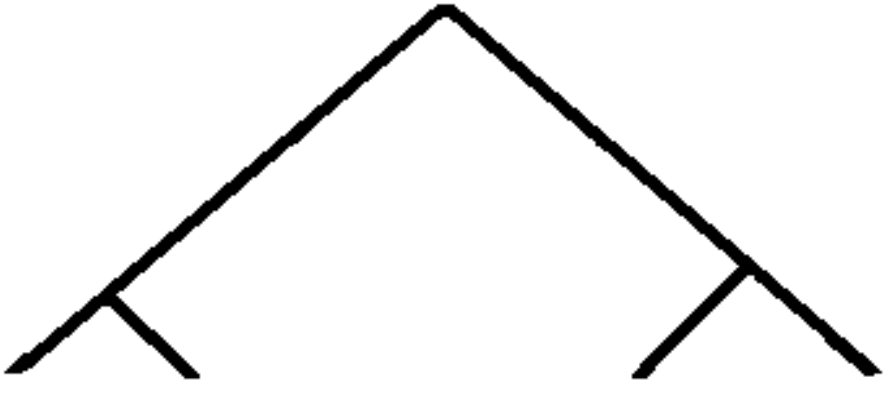}
\caption{
The two possible rooted quartets for a binary tree. Left: $K_{4}$ the four leaf rooted caterpillar tree 
(also known as a comb or pectinate tree),
right: $B_{4}$ the fully balanced tree on four leaves 
(also known as a fork, see e.g. \cite{BChoSSni2007} for some nomenclature).
\label{figTreeB4K4}
}
\end{figure}

Importantly for us in \cite{TCorAMirFRosGVal2018arXiv} it is shown in 
Lemma $4$ therein
that for $n>4$, the quartet index has a recursive representation as

\be\label{eqRecursQI}
QIB(T_{n}) = QIB(T_{L_{n}}) + QIB(T_{n-L_{n}}) + \binom{L_{n}}{2}\binom{n-L_{n}}{2},
\ee
where $T_{n}$ is the tree on $n$ leaves.

In \cite{TCorAMirFRosGVal2018arXiv} various models of tree growth were considered, Aldous's $\beta$--model, 
Ford's $\alpha$--model (\cite{DFor2005arXiv}, but see also \cite{TCorAMirFRos2018arXiv})
and Chen--Ford--Winkel's $\alpha$--$\gamma$--model \cite{BCheDForMWin2009}. In this work we will focus on the 
Aldous's $\beta\ge 0$--model of tree growth and characterize the limit distribution, as the number of 
leaves, $n$, grows to infinity, of the QI. We will take advantage of the 
recursive representation of Eq. \eqref{eqRecursQI} that allows for the usage of the
powerful contraction method. 

We require knowledge of the mean and variance of the QI 
for Aldous's $\beta$--model 
and these are (Corollaries $4$ and $7$ in \cite{TCorAMirFRosGVal2018arXiv}) 

\be\label{eqMeanVarQI}
\begin{array}{rcl}
\Expectation{QIB(T_{n})} & = & \frac{3\beta+6}{7\beta+18}\binom{n}{4}\\
\variance{QIB(T_{n})} & = & \frac{\left(\beta+2 \right)\left(2\beta^{2}+9\beta+12 \right)}{2\left(7\beta+18 \right)^{2}\left(127\beta^{3}+1383\beta^{2}+4958\beta+5880 \right)}n^{8} + O(n^{7}).
\end{array}
\ee

\section{Contraction method approach}\label{secCTdist}

Consider the space $D$ of distribution functions with finite second moment and first
moment equalling $0$. On $D$ we define the Wasserstein metric

$$
d(F,G) = \inf \Vert X-Y \Vert_{2}
$$
where $\Vert \cdot \Vert_{2}$ denotes the $L_{2}$ norm and the infimum is over all $X\sim F$, $Y\sim G$.
Notice that convergence in $d$ induces convergence in distribution.

Let $\tau \in [0,1]$ be a random variable whose distribution is not a Dirac $\delta$ at $0$ nor at $1$.
For $r\in \mathbb{N}_{+}$ define the transformation $S:D \to D$ by

\be\label{eqS}
S(F) = \mathcal{L}\left(\tau^{r}Y' + (1-\tau)^{r}Y'' + C(\tau) \right),
\ee
where $\mathcal{L}\left(X\right)$ denotes the law of the random variable $X$, $Y', Y'', \tau$ are independent, 
$Y', Y'' \sim F$,  $\tau \in [0,1]$; moreover we assume that $\tau$ satisfies, for all $n$,

\be\label{eqCondpni}
2\sum\limits_{i=1}^{n} p_{n,i} \left(\frac{i}{n}\right)^{2r} < 1,
\ee
where $p_{n,i} = P((i-1)/n < \tau \le  i/n)$
and the function $C(\cdot)$ is of the form

\be\label{eqCtau}
C(\tau) = \sum\limits_{r_{1}+r_{2}\le r}C_{r_{1},r_{2}}\tau^{r_{1}}(1-\tau)^{r_{2}}
\ee
for some constants $C_{r_{1},r_{2}}$ and furthermore satisfies $\E{C(\tau)}=0$.
By Thms. $3$ and $4$ in \cite{URos1992} $S$ is well defined, has a unique fixed point
and for any $F \in D$ the sequence $S^{n}(F)$ converges exponentially fast in the 
$d$ metric to $S$'s fixed point.
Using the exact arguments used to show  Thm. $2.1$ in \cite{URos1991} one can show that
the map $S$ is a contraction. Only the Lipschitz constant of convergence will differ
being $\sqrt{C_{\tau}}$, where $C_{\tau}=\max\{\E{\tau^{2r}},\E{(1-\tau)^{2r}}\}$ in our case. 
Notice that as $\tau\in[0,1]$ and is non--degenerate at the edges, then $C_{\tau}<1$ and we  have a
contraction.

We now state the main result of our work. We show weak convergence, with a characterization
of the limit for a class of recursively defined models.

\begin{theorem}[cf. Thm. $3.1$ in {\cite{URos1991}}]\label{thmYnYinfConv}
For $n\ge 2$, $\beta>0$ let $L_{n} \in \{1,\ldots,n-1\}$ be such that
$(L_{n}-1)$ is $\mathrm{BetaBinomial}(n-2,\beta+1,\beta+1)$
distributed and 
$\tau \sim \mathrm{Beta}(\beta+1,\beta+1)=:F_{\tau}$ distributed. 
Starting from the Dirac $\delta$ at $0$, i.e. $Y_{1}=0$ and convention
$\mathrm{BetaBinomial}(0,\beta+1,\beta+1)=\delta_{0}$,
for $r\in \mathbb{N}_{+}$ 
such that the condition of Eq. \eqref{eqCondpni} is met with the previous
choice of $F_{\tau}$, define recursively the sequence of random variables,

$$
Y_{n} = \left(\frac{L_{n}}{n}\right)^{r}Y_{L_{n}}+\left(1-\frac{L_{n}}{n}\right)^{r}Y_{n-L_{n}}+C_{n}(L_{n}),
$$
where the function $C_{n}(\cdot)$ is of the form

\be\label{eqCni}
C_{n}(i)=n^{-r}\left(
\sum\limits_{r_{1}+r_{2}+r_{3}\le r}C_{r_{1},r_{2},r_{3}}i^{r_{1}}(n-i)^{r_{2}}n^{r_{3}}+h_{n}(i) \right),
\ee
where $\Expectation{C_{n}(L_{n})}=0$ and $\sup_{i}n^{-r}h_{n}(i)\to 0$.
If $\E{Y_{n}^{2}}$ is uniformly bounded then, the random variable $Y_{n}$ converges in the Wasserstein $d$--metric to the random variable
$Y_{\infty}$ whose distribution satisfies the unique fixed point of $S$ (Eq. \ref{eqS}).
\end{theorem}
Notice that as $Y_{1}=0$ and by the definition of the recursion we will have $\E{Y_{n}}=0$ for all $n$.

The Yule tree case will be the limit of $\beta=0$ and this case the proof of the result
will be more straightforward (as commented on in the proof of Thm. \ref{thmYnYinfConv}).

Notice that $L_{n}/n \stackrel{D}{\to} \tau $.
It would be tempting to suspect that Thm. \ref{thmYnYinfConv} should 
be the conclusion of a general result related to the contraction method
(as presented in Eq. $(8.12)$, p. $351$ in \cite{MDrm2009}).
However, to the best of my knowledge, general results assume $L_{2}$ convergence 
of $L_{n}/n$ (e.g. Thm. $8.6$, p. $354$ in \cite{MDrm2009}), while in our phylogenetic
balance index case we will have only convergence in distribution. In such
a case it seems that convergence has to be proved case by case
(e.g. examples in \cite{SRacLRus1995}). Here we show the convergence
of Thm. \ref{thmYnYinfConv} similarly as in \cite{URos1991}.

We first derive a lemma that controls the non--homogeneous part of the 
recursion, i.e. $C_{n}(\cdot)$ as defined in Eq. \eqref{eqCni}.

\begin{lemma} [cf. Prop. $3.2$ in {\cite{URos1991}}]\label{lemProp3.2}
Let $C_{n}:\{ 1, \ldots, n-1 \} \to \mathbb{R}$ be as in Eq. \eqref{eqCni}. Then

\be 
\sup\limits_{x \in [0,1)} \biggr \vert C_{n}(\lfloor (n-1)x \rfloor+1) - C(x)\biggr  \vert \le \sup_{i}n^{-r}h_{n}(i) + O(n^{-1}).
\ee
\end{lemma}
\begin{proof}
For $1 \le \lfloor (n-1)x \rfloor +1 \le n-1$ and writing $i=\lfloor (n-1)x \rfloor+1$ we have 
due to the representation of Eqs. \eqref{eqCtau} and \eqref{eqCni}

$$
\begin{array}{l}
\biggr \vert C_{n}(\lfloor (n-1)x \rfloor+1) - C(x)\biggr  \vert 
\le \max\{C_{r_{1},r_{2}} \} \left(
\biggr \vert \left(\frac{i}{n} \right)^{r} -x^{r}\biggr  \vert +
\right. \\ \left.
\biggr \vert \left(1-\frac{i}{n} \right)^{r} -(1-x)^{r} \biggr \vert +
\sum\limits_{r_{1}+r_{2}\le r} \biggr \vert
\left(\frac{i}{n} \right)^{r_{1}}\left(1-\frac{i}{n} \right)^{r_{2}} -x^{r_{1}}(1-x)^{r_{2}}
 \right) 
\biggr \vert
\\ 
+ \sup\limits_{i}n^{-r}h_{n}(i).
\end{array}
$$
Bounding the individual components, using the mean value theorem and that by
construction $x$ cannot differ from $i/n$ by more than $1/n$ we have

$$
\begin{array}{l}
\biggr \vert \left(\frac{i}{n} \right)^{r} -x^{r} \biggr \vert 
\le r \biggr\vert \frac{i}{n} -x \biggr\vert \le \frac{r}{n}
 = O(n^{-1})
\end{array}
$$
and

$$
\begin{array}{l}
\biggr \vert \left(1-\frac{i}{n} \right)^{r} -(1-x)^{r} \biggr \vert 
\le r \biggr\vert \frac{i}{n} -x \biggr\vert \le \frac{r}{n}
= O(n^{-1}).
\end{array}
$$
Furthermore, immediately by the triangle inequality and the two above inequalities

$$
\begin{array}{l}
\biggr \vert \left(\frac{i}{n} \right)^{r_{1}}\left(1-\frac{i}{n} \right)^{r_{2}} -x^{r_{1}}(1-x)^{r_{2}}\biggr \vert 
= O(n^{-1}).
\end{array}
$$
\end{proof}

\begin{lemma}[cf. Prop. $3.3$ in {\cite{URos1991}}]\label{lemProp3.3}
Let $a_{n}$, $b_{n}$, $p_{n,i}$, $n\in \mathbb{N}$ be three sequences such that
$0 \le b_{n} \to 0$ with $n$,  $0\le p_{n,i} \le 1$, 

\be \label{eqanconvpn}
0 \le a_{n+1} \le 2 \sum\limits_{i=1}^{n}p_{n,i}\left(\frac{i}{n}\right)^{R}\left(\sup\limits_{i \in \{1,\ldots,n\}}a_{i} \right) + b_{n}.
\ee
and

$$
0 < 2\sum\limits_{i=1}^{n}p_{n,i}\left(\frac{i}{n}\right)^{R} = C < 1.
$$
Then $\lim_{n\to \infty} a_{n} = 0$.
\end{lemma}
\begin{proof}
The proof is exactly the same as the  proof of Proposition $3.3$ in \cite{URos1991}.
In the last step
we will have with $a:=\limsup a_{n} < \infty$ the sandwiching for all $\epsilon>0$

$$
0 \le a \le C(a+\epsilon).
$$
\end{proof}

Having Lemmata \ref{lemProp3.2} and \ref{lemProp3.3} we turn to showing 
Thm. \ref{thmYnYinfConv}.

\begin{proof}[Proof of Thm. \ref{thmYnYinfConv}]
Denote the law of $Y_{n}$ as $\mathcal{L}(Y_{n})=G_{n}$.
We take $Y_{\infty}$ and $Y_{\infty}'$ independent and
distributed as $G_{\infty}$, the fixed point of $S$. Then, for $i=1,\ldots,n-1$
we choose independent versions of $Y_{i}$ and $Y_{i}'$.
We need to show $d^{2}(G_{n},G_{\infty}) \to 0$. As the metric is the infimum over all pairs of random
variables that have marginal distributions $G_{n}$ and $G_{\infty}$ the obvious choice is 
to take $Y_{n}$, $Y_{\infty}$ such that $L_{n}/n$ will be close to $\tau$ for large $n$. 
The Yule model $(\beta=0)$   was considered in \cite{URos1991}
and there $\tau \sim \mathrm{Unif}[0,1]$
and $L_{n}$ is uniform on $\{1,\ldots,n-1\}$. Hence, $\lfloor (n-1)\tau \rfloor + 1$ will
be uniform on $\{1,\ldots,n-1\}$, remember $P(\tau=1)=0$, and $L_{n}/n  \stackrel{D}{=} (\lfloor (n-1)\tau \rfloor + 1)/n$.
However, when $\beta > 0$ the situation complicates. For a given $n$, $(L_{n}-1)$ is BetaBinomial$(n-2,\beta+1,\beta+1)$ 
distributed (cf. Eq. \ref{eqqni} and  Eqs. $1$ and $3$ in \cite{DAld1996}). Hence, if 
$\tau \sim \mathrm{Beta}(\beta+1,\beta+1)$ and $(L_{n}-1) \sim \mathrm{BetaBinomial}(n-2,\beta+1,\beta+1)$
we do not have $L_{n}/n  \stackrel{D}{=} (\lfloor (n-1)\tau \rfloor + 1)/n$ exactly.
We may bound the Wasserstein metric by any coupling that retains the marginal distributions
of the two random variables. Therefore, from now on we will be considering a version, where conditional on $\tau$,
the random variable $(L_{n}-1)$ is Binomial$(n-2,\tau)$ distributed. 
Let $r_{n}$ be any sequence such that $r_{n}/n \to 0$ and $n/r_{n}^{2} \to 0$, e.g.
$r_{n} = n \ln^{-1} n$. Then, by Chebyshev's inequality

$$
P\left( \vert L_{n} - \Expectation{L_{n} \vert \tau} \vert  \ge r_{n} \biggr\vert \tau \right) \le \frac{n\tau(1-\tau)}{r_{n}^{2}} \le \frac{n}{4r_{n}^{2}} \to 0.
$$
We now want to show
$d^{2}(G_{n},G_{\infty})\to 0$ 
and we will exploit the above coupling in the bound

$$
\begin{array}{l}
d^{2}(G_{n},G_{\infty}) \le
\E{\left(\left(\left(\frac{L_{n}}{n} \right)^{r}Y_{L_{n}} - \tau^{r}Y_{\infty}\right)
+\left(\left(\frac{n-L_{n}}{n} \right)^{r}Y_{n-L_{n}} - (1-\tau)^{r}Y'_{\infty}\right)
\right. \right. \\ \left. \left.
+\left(C_{n}(L_{n})-C(\tau)\right)\right)^{2}
}
=
\E{\left(\left(\frac{L_{n}}{n} \right)^{r}Y_{L_{n}} - \tau^{r}Y_{\infty}\right)^{2}}
\\
+\E{\left(\left(\frac{n-L_{n}}{n} \right)^{r}Y_{n-L_{n}} - (1-\tau)^{r}Y'_{\infty}\right)^{2}}
+\E{\left(C_{n}(L_{n})-C(\tau)\right)^{2}},
\end{array}
$$
where $Y_{\infty},Y'_{\infty} \sim G_{\infty}$ are independent.
Remember that $\E{Y_{i}}=\E{Y_{\infty}}=0$ so that the expectation of the cross products disappears.

Our main step is to have a bound where the $L_{n}/n$ term is replaced by some transformation of $\tau$.
Let $\tilde{r}_{n}$ be a (appropriate) random integer in $\{\pm 1,\ldots,\pm \lceil r_{n} \rceil\}$
and we may write (with the chosen coupling of $L_{n}$ and $\tau$),

$$
\begin{array}{l}
\E{
\left(\left(\frac{L_{n}}{n} \right)^{r}Y_{L_{n}} - \tau^{r}Y_{\infty}\right)^{2}}
=
\E{\Expectation{
\left(\left(\frac{L_{n}}{n} \right)^{r}Y_{L_{n}} - \tau^{r}Y_{\infty}\right)^{2}\biggr \vert \tau }}
\\
= 
\E{\E{
\left(\left(\frac{\lfloor (n-1)\tau \rfloor +1 + \tilde{r}_{n}}{n} \right)^{r}Y_{L_{n}} - \tau^{r}Y_{\infty}\right)^{2}
\biggr\vert  \vert L_{n} - \E{L_{n}} \vert \le r_{n}, \tau
}
 \right. \\ \left. \cdot
P(\vert L_{n} - \E{L_{n}} \vert \le r_{n} \biggr \vert \tau)  }
\\ +
\E{\E{
\left(\left(\frac{L_{n}}{n} \right)^{r}Y_{L_{n}} - \tau^{r}Y_{\infty}\right)^{2}
\biggr \vert \vert L_{n} - \E{L_{n}} \vert \ge r_{n}, \tau
}
 \right. \\ \left. \cdot
P(\vert L_{n} - \E{L_{n}} \vert \ge r_{n} \biggr \vert \tau)  }
\\ \le 
\E{\left(\left(\frac{\lfloor (n-1)\tau \rfloor + 1+ \tilde{r}_{n}}{n} \right)^{r}Y_{L_{n}} - \tau^{r}Y_{\infty}\right)^{2}
}+
\frac{n}{4r_{n}^{2}}\E{\left(\left(\frac{L_{n}}{n} \right)^{r}Y_{L_{n}} - \tau^{r}Y_{\infty}\right)^{2}}
\\ =
\E{\left(
\left(
\left(\frac{\lfloor (n-1)\tau \rfloor +1}{n} \right)^{r}
+r \frac{\tilde{r}_{n}}{n}\left(\frac{\lfloor (n-1)\tau \rfloor +1 + \xi_{\tilde{r}_{n}}}{n} \right)^{r-1}
\right)Y_{L_{n}} - \tau^{r}Y_{\infty}\right)^{2}}
\\+
\frac{n}{4r_{n}^{2}}\E{\left(\left(\frac{L_{n}}{n} \right)^{r}Y_{L_{n}} - \tau^{r}Y_{\infty}\right)^{2}}
%
=
\E{\left(\left(\frac{\lfloor (n-1)\tau \rfloor+1}{n} \right)^{r}Y_{L_{n}} - \tau^{r}Y_{\infty}\right)^{2}}
\\
+r^{2}n^{-2}
\E{\tilde{r}_{n}^{2}\left(\frac{\lfloor (n-1)\tau \rfloor +1 + \xi_{\tilde{r}_{n}}}{n} \right)^{2(r-1)}Y_{L_{n}}^{2}}
\\ +2r n^{-1} \E{\tilde{r}_{n}\left(\frac{\lfloor (n-1)\tau \rfloor + 1+\xi_{\tilde{r}_{n}}}{n} \right)^{r-1}Y_{L_{n}}\cdot
\left(\left(\frac{\lfloor (n-1)\tau \rfloor +1 }{n} \right)^{r}Y_{L_{n}} - \tau^{r}Y_{\infty}\right)}
\\ +
\frac{n}{4r_{n}^{2}}\E{\left(\left(\frac{L_{n}}{n} \right)^{r}Y_{L_{n}} - \tau^{r}Y_{\infty}\right)^{2}},
\end{array}
$$ 
where $\xi_{\tilde{r}_{n}} \in (0,\tilde{r}_{n})$ 
is (a random variable) such that the mean value theorem holds (for the function $(\cdot)^{r}$).
As $Y_{n}$, $Y_{\infty}$ have uniformly bounded second moments and 
$0 \le \xi_{\tilde{r}_{n}} \le \tilde{r}_{n} \le r_{n} \le n$ 
we have, by the assumptions $r_{n} /n \to 0$ and $n/r^{2}_{n}\to 0$, 

$$
\begin{array}{l}
\left(r\frac{r_{n}}{n} \right)^{2}
\E{\left(\frac{\lfloor (n-1)\tau \rfloor + 1+ \xi_{\tilde{r}_{n}}}{n} \right)^{2(r-1)}Y_{L_{n}}^{2}}
+
\frac{n}{4r_{n}^{2}}\E{\left(\left(\frac{L_{n}}{n} \right)^{r}Y_{L_{n}} - \tau^{r}Y_{\infty}\right)^{2}}
\\ 
+2r \frac{r_{n}}{n}\E{\left(\frac{\lfloor (n-1)\tau \rfloor +1 + \xi_{\tilde{r}_{n}}}{n} \right)^{r-1}Y_{L_{n}} \cdot
\left(\left(\frac{\lfloor (n-1)\tau \rfloor +1}{n} \right)^{r}Y_{L_{n}} - \tau^{r}Y_{\infty}\right)}
\to 0
\end{array}
$$
and hence for some sequence $u_{n}\to 0$ we have, 

$$
\begin{array}{l}
\E{\left(\left(\frac{L_{n}}{n} \right)^{r}Y_{L_{n}} - \tau^{r}Y_{\infty}\right)^{2}}
\le 
\E{\left(\left(\frac{\lfloor (n-1)\tau \rfloor+1}{n} \right)^{r}Y_{L_{n}} - \tau^{r}Y_{\infty}\right)^{2}}
+ u_{n}.
\end{array}
$$

Remembering the assumption $\sup_{i} n^{-r} h_{n}(i) \to 0$, 
the other component can be treated in the same way as 
$\E{\left(\left(\frac{L_{n}}{n} \right)^{r}Y_{L_{n}} - \tau^{r}Y_{\infty}\right)^{2}}$ with conditioning 
on $\tau$ and then controlling by $r_{n}$ and Chebyshev's inequality $L_{n}$'s deviation from its expected value. 
We therefore have for some sequence $v_{n} \to 0$

$$
\begin{array}{l}
d^{2}(G_{n},G_{\infty})  \le 
\E{\left(\left(\frac{\lfloor (n-1)\tau \rfloor +1}{n} \right)^{r}Y_{L_{n}} - \tau^{r}Y_{\infty}\right)^{2}}
\\
+\E{\left(C_{n}(\lfloor (n-1)\tau \rfloor+1) - C(\tau)\right)^{2}}
\\
+ \E{\left(\left(\frac{n-\lfloor (n-1)\tau \rfloor -1}{n} \right)^{r}Y_{n-L_{n}} - (1-\tau)^{r}Y_{\infty}'\right)^{2}}
+ v_{n}.
\end{array}
$$
Consider the first term of the right--hand side of the inequality
and denote $d^{2}_{n-1} := \sup_{i \in \{1,\ldots,n-1\}}d^{2}(G_{i},G_{\infty})$

$$
\begin{array}{l}
\E{\left(\left(\frac{\lfloor (n-1)\tau \rfloor +1}{n} \right)^{r}Y_{L_{n}} - \tau^{r}Y_{\infty}\right)^{2}}
\\
=\E{\sum\limits_{i=1}^{n-1}1_{(i-1)/(n-1)<\tau \le i/(n-1)} \left(\left(\frac{i}{n}\right)^{r} Y_{L_{n}}-\tau^{r}Y_{\infty}\right)^{2}}
\\ \le \sum\limits_{i=1}^{n-1}p_{n-1,i} \left(\frac{i}{n}\right)^{2r}\E{\left( Y_{L_{n}}-Y_{\infty}\right)^{2}}
= \sum\limits_{i=1}^{n-1}p_{n-1,i} \left(\frac{i}{n}\right)^{2r}d^{2}_{n-1},
\end{array}
$$
where $p_{n,i}=P( (i-1)/(n-1) < \tau \le i/(n-1))$.
Invoking Lemmata \ref{lemProp3.2}, \ref{lemProp3.3} 
and using the assumption of Eq. \eqref{eqCondpni} with $R=2r$ we have 

$$
\begin{array}{l}
d^{2}(G_{n},G_{\infty}) \le  
2\sum\limits_{i=1}^{n-1}p_{n-1,i} \left(\frac{i}{n}\right)^{2r} d^{2}_{n-1}
+
\left(n^{-r} \sup\limits_{i}h_{n}(i)\right)^{2} +v_{n} + O(n^{-2})
\end{array}
$$
which converges to $0$.

\end{proof}

\section{Limit distribution of the quartet index for Aldous's $\beta\ge 0$--model trees}\label{secLimQI}
We show here that the QIB of Aldous's $\beta\ge 0$--model trees satisfies the conditions of
Thm. \ref{thmYnYinfConv} with $r=4$ and hence the QIB has a well characterized limit
distribution. We define a centred and scaled version of the QIB for Aldous's $\beta\ge 0$--model 
tree on $n \ge 4$ leaves

\be
Y_{n}^{Q} = n^{-4} \left( QIB(T_{n}) - \frac{3\beta+6}{7\beta+18}\binom{n}{4} \right).
\ee
We now specialize Thm. \ref{thmYnYinfConv} to the QIB case and assume $Y_{1}=Y_{2}=Y_{3}=0$ for completeness

\begin{theorem}
The sequence of random variables $Y^{Q}_{n}$
for trees generated by Aldous's $\beta$--model with $\beta\ge 0$
converges with $n\to \infty$ in the Wasserstein d--metric (and hence in distribution) to a
random variable $Y_{Q} \sim \mathcal{Q} \equiv G_{\infty}$ satisfying the following equality in distribution

\be
\begin{array}{l}
Y_{Q} \stackrel{\mathcal{D}}{=} \tau^{4} Y_{Q}' + (1-\tau)^{4}Y_{Q}'' 
+ \frac{3\beta+6}{24(7\beta+18)}\left(\tau^{4} +\left(1-\tau\right)^{4} \right)
- \frac{3\beta+6}{24(7\beta+18)}
\\
+ \frac{1}{4}\tau^{2}(1-\tau)^{2},
\end{array}
\ee
where $\tau\sim F_{\tau}$ is distributed as the Beta distribution of Eq. \eqref{eqBtreedist}, 
$Y_{Q}, Y_{Q}', Y_{Q}'' \sim \mathcal{Q}$
and $Y_{Q}', Y_{Q}'', \tau$ are all independent.
\end{theorem}
\begin{proof}
Denote by $P_{3}(x,y)$ a polynomial of degree at most three in terms
of the variables $x$, $y$.
From the recursive representation of Eq. \eqref{eqRecursQI} for $n>4$

$$
\begin{array}{l}
Y_{n}^{Q} = n^{-4}\left(QIB(T_{L_{n}}) -\frac{3\beta+6}{(7\beta+18)}\binom{L_{n}}{4} + QIB(T_{n-L_{n}}) - \frac{3\beta+6}{(7\beta+18)}\binom{n-L_{n}}{4}
\right. \\ \left.
+ \binom{L_{n}}{2}\binom{n-L_{n}}{2}
 + \frac{3\beta+6}{(7\beta+18)}\binom{L_{n}}{4} + \frac{3\beta+6}{(7\beta+18)}\binom{n-L_{n}}{4}  - \frac{3\beta+6}{(7\beta+18)}\binom{n}{4} \right)
\\ = 
\left(\frac{L_{n}}{n}\right)^{4}Y_{L_{n}}^{Q} + \left(1-\frac{L_{n}}{n}\right)^{4}Y_{n-L_{n}}^{Q}
+ \frac{1}{4}\left(\frac{L_{n}}{n}\right)^{2}\left(1-\frac{L_{n}}{n}\right)^{2} 
+ \frac{3\beta+6}{24(7\beta+18)}\left(\frac{L_{n}}{n}\right)^{4}
\\
+ \frac{3\beta+6}{24(7\beta+18)}\left(1-\frac{L_{n}}{n}\right)^{4}
- \frac{3\beta+6}{24(7\beta+18)}+ n^{-4}P_{3}(n,L_{n}).
\end{array}
$$
We therefore have $r=4$ and 

$$
\begin{array}{l}
C_{n}(i) = \frac{1}{4}\left(\frac{i}{n}\right)^{2}\left(1-\frac{i}{n}\right)^{2} 
+ \frac{3\beta+6}{24(7\beta+18)}\left( \left(\frac{i}{n}\right)^{4} +\left(1-\frac{i}{n}\right)^{4} \right)
 - \frac{3\beta+6}{24(7\beta+18)}
\\+ n^{-4}P_{3}(n,i).
\end{array}
$$
By the scaling and centring we know that $E{Y_{n}^{Q}}=0$ and 
$E{\left(Y_{n}^{Q}\right)^{2}}$ is uniformly bounded by 
Eq. \eqref{eqMeanVarQI}.
For the Beta law of $\tau$ we need to examine for all $i$

$$
p_{n,i}:=P(\frac{i-1}{n} \le \tau < \frac{i}{n}) = 
\frac{\Gamma(2\beta+2)}{\Gamma^{2}(\beta+1)}
\int\limits_{(i-1)/n}^{i/n}x^{\beta}(1-x)^{\beta} \ud x.
$$
We consider two cases
\begin{enumerate}
\item $\beta>0$, we have to check if the condition of Eq. \eqref{eqCondpni} is satisfied.
Let 

$$
B_{x}(\beta+1,\beta+1) = \int\limits_{0}^{x} u^{\beta}(1-u)^{\beta} \ud u
$$
be the incomplete Beta function. Then,

$$
\begin{array}{rcl}
p_{n,i} & = & \frac{\Gamma(2\beta+2)}{\Gamma^{2}(\beta+1)}\left( B_{i/n}(\beta+1,\beta+1) - B_{(i-1)/n}(\beta+1,\beta+1) \right)
\\ &= &n^{-1} \frac{\Gamma(2\beta+2)}{\Gamma^{2}(\beta+1)}B_{\xi}'(\beta+1,\beta+1)
\end{array}
$$
by the mean value theorem for some $\xi \in ((i-1)/n,i/n)$. Obviously

$$
B_{\xi}'(\beta+1,\beta+1) = \xi^{\beta}(1-\xi)^{\beta} \le \left(\frac{i}{n}\right)^{\beta}\left(1-\frac{i-1}{n}\right)^{\beta}
$$
and now 

\be \label{eqpnibound}
\begin{array}{l}
\sum\limits_{i=1}^{n}p_{n,i} \left(\frac{i}{n} \right)^{R} 
\le \frac{\Gamma(2\beta+2)}{\Gamma^{2}(\beta+1)}
n^{-1}  \sum\limits_{i=1}^{n} \left(\frac{i}{n} \right)^{R}  \left(\frac{i}{n}\right)^{\beta}\left(1-\frac{i-1}{n}\right)^{\beta}
\\ \to \frac{\Gamma(2\beta+2)}{\Gamma^{2}(\beta+1)}\int\limits_{0}^{1}u^{\beta+R}(1-u)^{\beta} \ud u
= \frac{\Gamma(2\beta+2)}{\Gamma^{2}(\beta+1)} \frac{\Gamma(\beta+R+1)\Gamma(\beta+1)}{\Gamma(2\beta+R+2)}
\\ = \frac{\Gamma(2\beta+2)}{\Gamma(\beta+1)} \frac{\Gamma(\beta+R+1)}{\Gamma(2\beta+R+2)}.
\end{array}
\ee
Take $1<R_{1}<R_{2}$ and consider the ratio

$$
\begin{array}{l}
A=\frac{\Gamma(2\beta+2)}{\Gamma(\beta+1)} \frac{\Gamma(\beta+R_{2}+1)}{\Gamma(2\beta+R_{2}+2)}
\left(
\frac{\Gamma(2\beta+2)}{\Gamma(\beta+1)} \frac{\Gamma(\beta+R_{1}+1)}{\Gamma(2\beta+R_{1}+2)}
\right)^{-1}
\\~\\
= \frac{\Gamma(\beta+R_{2}+1)}{\Gamma(2\beta+R_{2}+2)} \frac{\Gamma(\beta)}{\Gamma(\beta)} \frac{\Gamma(2\beta+R_{1}+1)}{\Gamma(\beta+R_{1}+1)}
= \frac{B(\beta+1+R_{2},\beta)}{B(\beta+1+R_{1},\beta)}.
\end{array}
$$
The ratio $A<1$ as the Beta function is decreasing in its arguments---hence the derived
upper bound in Eq. \eqref{eqpnibound} is decreasing in $R$. For $R=1$ the bound equals

$$
\frac{\Gamma(2\beta+2)}{\Gamma(\beta+1)} \frac{\Gamma(\beta+2)}{\Gamma(2\beta+3)}
= 
\frac{\Gamma(2\beta+2)}{\Gamma(\beta+1)} \frac{(\beta+1)\Gamma(\beta+1)}{(2\beta+2)\Gamma(2\beta+2)}
= \frac{(\beta+1)}{2(\beta+1)}=\frac{1}{2}
$$
and hence for all $R>1$ and all $\beta>0$

$$
\sum\limits_{i=1}^{n}p_{n,i} \left(\frac{i}{n} \right)^{R}  < \frac{1}{2}.
$$
As in our case we have $r\ge 1$, then for $R=2r\ge 2$
the assumptions of Lemma \ref{lemProp3.3} are satisfied and the statement of the theorem follows through.
\item $\beta=0$, then directly $p_{n,i}=n^{-1}$, Eq. \eqref{eqCondpni} and assumptions
of Lemma \ref{lemProp3.3} are immediately satisfied and the statement of the theorem follows through.
This is the Yule model case, in which the proof of the counterpart of  Thm. \ref{thmYnYinfConv}
is much more straightforward, as mentioned before.
\end{enumerate}
\end{proof}

\begin{remark}
When $\beta <0$ the process $L_{n}/n$ seems to have a more involved
asymptotic behaviour (cf. Lemma $3$ in \cite{DAld1996} in the $\beta \le -1$ case). 
Furthermore, the bounds applied here do not hold for $\beta<0$. 
Therefore, this family of tree models (including the
important uniform model, $\beta=-3/2$) deserves a separate study with respect to its
quartet index.
\end{remark}

\section{Comparing with simulations}\label{secSimul}
To verify the results we compared the simulated values from the limiting theoretical distribution
of $Y_{Q}$ with scaled and centred values of Yule tree QI values. The $500$--leaf Yule trees were simulated
using the \texttt{rtreeshape()} function of the apTreeshape \cite{NBorEDurMBluOFra2012} 
R \cite{R} package and Tom\'as Mart\'inez--Coronado's 
in--house Python code. Then, for each tree the QI value was calculated by 
Gabriel Valiente's and Tom\'as Mart\'inez--Coronado's
in--house programs. The raw values $QIB(\mathrm{Yule}_{500})$ were scaled and centred as

$$
Y^{Q}_{n} = 500^{-4} \left(QIB(\mathrm{Yule}_{500}) - \frac{1}{3}\binom{500}{4} \right).
$$
The $Y_{Q}$ values were simulated using the proposed in \cite{KBar2018} heuristic Algorithm $3$
(R code in Appendix). The results of the simulation are presented in Fig. \ref{figSimQIYule}.

\begin{figure}[!ht]
\centering
\includegraphics[width=0.48\textwidth]{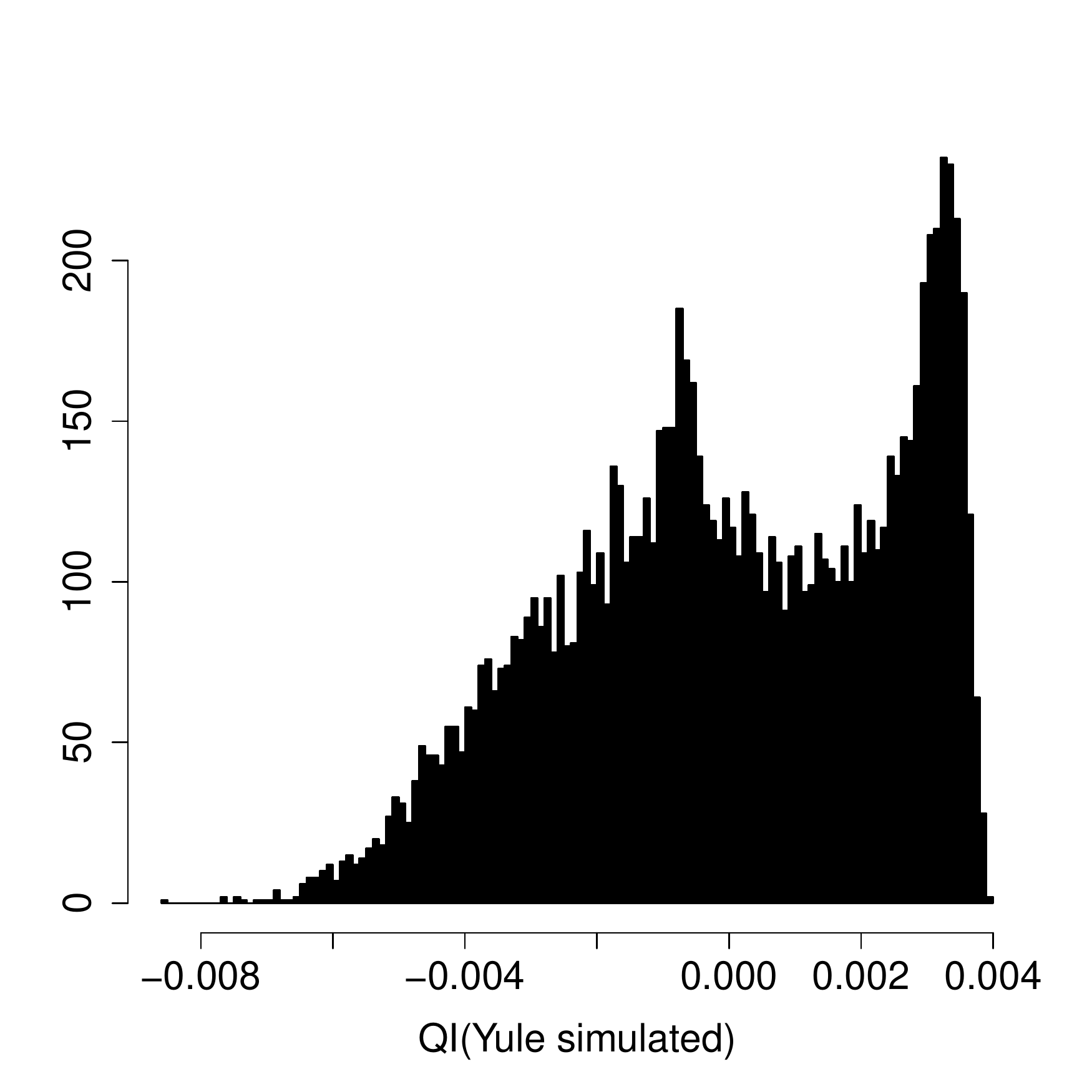}
\includegraphics[width=0.48\textwidth]{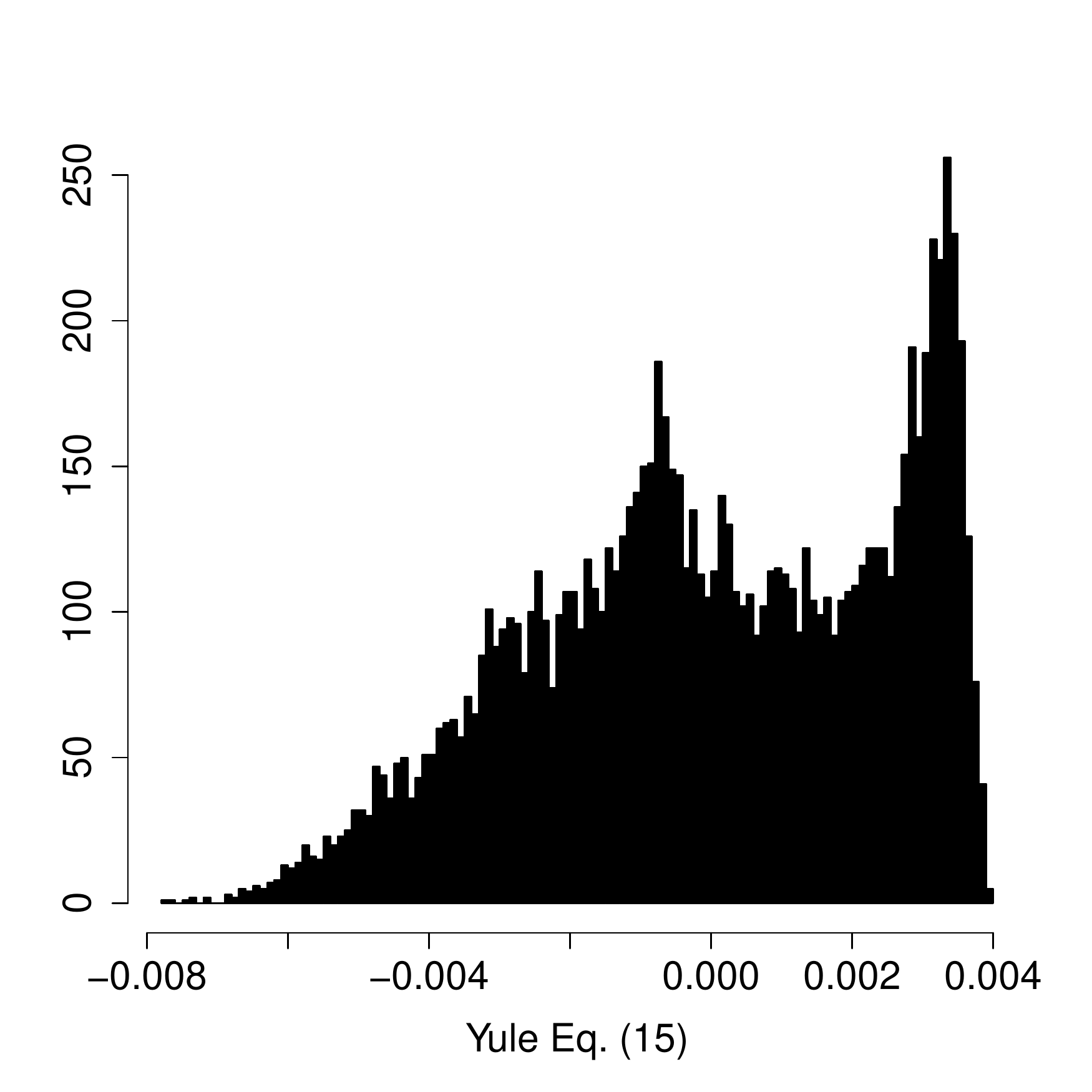}
\caption{
Left: histogram of scaled and centred simulated values of the QIB for the Yule tree, $Y^{Q}_{n}$,
right: histogram of $Y_{Q}$ for the Yule model, $\beta=0$. The mean, variance, skewness and excess kurtosis
of the simulated values are
$-3.177\cdot10^{-6}$, $6.321\cdot 10^{-6}$, $-0.308$, $-0.852$ (left, simulated values)
and 
$1.682\cdot 10^{-5}$, $6.38\cdot 10^{-6}$, $-0.317$, $-0.834$ (right, theoretical values
of the  heuristic Algorithm $3$ in \cite{KBar2018} with recursion depth $15$).
For $\beta=0$ the leading constant of the variance in Eq. \eqref{eqMeanVarQI} is 
$5/(24\cdot 33075)\approx 6.299\cdot 10^{-6}$.
\label{figSimQIYule}
}
\end{figure}

\subsection*{Acknowledgements}
I would like to thank the whole 
Computational Biology and Bioinformatics Research Group of the Balearic Islands University
for hosting me on multiple occasions, introducing me to problems related with the
quartet index and for valuable comments on this manuscript. 
The simulated values of the quartet index for the Yule tree were provided
by Gabriel Valiente and Tom\'as Mart\'inez--Coronado.
I was supported by the Knut and Alice Wallenberg Foundation
and am by the Swedish Research Council (Vetenskapsr\aa det) grant no. $2017$--$04951$.
My collaboration with the Balearic Islands University was partially supported by the 
the G S Magnuson Foundation of the Royal Swedish Academy of Sciences
(grants no. MG$2015$--$0055$, MG$2017$--$0066$)
and The Foundation for Scientific Research and Education in Mathematics (SVeFUM).
I am grateful to anonymous reviewers for comments that significantly improved 
this work.

\bibliographystyle{plainnat}
\bibliography{QIindex_arXiv.bib}
\clearpage
\section*{Appendix: R code for simulating from the limit distribution of the normalized quartet index}
\begin{lstlisting}
fCtau_QIB<-function(x,beta=0){
    PB4beta<-(3*beta+6)/(7*beta+18);
    PB4beta*(x^4)/24+PB4beta*((1-x)^4)/24-PB4beta/24+0.25*x*x*(1-x)^2
} 

fdistribution_limitQIB<-function(num.iter=10,popsize=10000,Y0=0){
    replicate(popsize,fdraw_limitQIB(num.iter,Y0))
}

fdraw_limitQIB<-function(num.iter=15,Y0=0){
    res<-0
    if (num.iter==0){
	Y1<-Y0
	Y2<-Y0
    }
    else{
	Y1<-fdraw_limitQIB(num.iter-1,Y0)
	Y2<-fdraw_limitQIB(num.iter-1,Y0)	
    }
    tau<-runif(1)
    res<-((tau)^4)*Y1+((1-tau)^4)*Y2+fCtau_QIB(tau)
    res
}

popsize<-10000 ## size of sample for histogram
num.iter<-15 ## depth of the recursion
Y0 <- 0 ## initial value

vlimitQIB<-fdistribution_limitQIB(num.iter=num.iter,popsize=popsize)


\end{lstlisting}

\end{document}